\documentclass{wscpaperproc}
\usepackage{latexsym}
\usepackage{graphicx}
\usepackage{mathptmx}
\usepackage[T1]{fontenc}
\usepackage{mwe}

\usepackage{amsmath}
\usepackage{amsfonts}
\usepackage{amssymb}
\usepackage{amsbsy}
\usepackage{amsthm}
\usepackage[ruled,vlined]{algorithm2e}

\usepackage[pdftex,colorlinks=true,urlcolor=blue,citecolor=black,anchorcolor=black,linkcolor=black]{hyperref}

\newtheoremstyle{wsc}
{3pt}
{3pt}
{}
{}
{\bf}
{}
{.5em}
{}

\theoremstyle{wsc}
\newtheorem{theorem}{Theorem}

\begin{document}

%
%

\pagestyle{fancyplain}

\thispagestyle{plain}
\firstPageHead{}

\chead{\fancyplain{}{\itshape Aldridge}}

\rhead{}
\cfoot{}
\renewcommand{\headrulewidth}{0pt} 

\makeatletter
\let\@internalcite\cite
\def\cite{\def\@citeseppen{-1000}%
    \def\@cite##1##2{(##1\if@tempswa , ##2\fi)}%
    \def\citeauthoryear##1##2##3{##1 ##3}\@internalcite}
\def\citeNP{\def\@citeseppen{-1000}%
    \def\@cite##1##2{##1\if@tempswa , ##2\fi}%
    \def\citeauthoryear##1##2##3{##1 ##3}\@internalcite}
\def\citeN{\def\@citeseppen{-1000}%
    \def\@cite##1##2{##1\if@tempswa, ##2)\else{}\fi}%
    \def\citeauthoryear##1##2##3{##1 (##3)}\@citedata}
\def\citeA{\def\@citeseppen{-1000}%
    \def\@cite##1##2{(##1\if@tempswa , ##2\fi)}%
    \def\citeauthoryear##1##2##3{##1}\@internalcite}
\def\citeANP{\def\@citeseppen{-1000}%
    \def\@cite##1##2{##1\if@tempswa , ##2\fi}%
    \def\citeauthoryear##1##2##3{##1}\@internalcite}
\def\shortcite{\def\@citeseppen{-1000}%
    \def\@cite##1##2{(##1\if@tempswa , ##2\fi)}%
    \def\citeauthoryear##1##2##3{##2 ##3}\@internalcite}
\def\shortciteNP{\def\@citeseppen{-1000}%
    \def\@cite##1##2{##1\if@tempswa , ##2\fi}%
    \def\citeauthoryear##1##2##3{##2 ##3}\@internalcite}
\def\shortciteN{\def\@citeseppen{-1000}%
    \def\@cite##1##2{##1\if@tempswa, ##2\else{}\fi}%
    \def\citeauthoryear##1##2##3{##2 (##3)}\@citedata}
\def\shortciteA{\def\@citeseppen{-1000}%
    \def\@cite##1##2{(##1\if@tempswa , ##2\fi)}%
    \def\citeauthoryear##1##2##3{##2}\@internalcite}
\def\shortciteANP{\def\@citeseppen{-1000}%
    \def\@cite##1##2{##1\if@tempswa , ##2\fi}%
    \def\citeauthoryear##1##2##3{##2}\@internalcite}
\def\citeyear{\def\@citeseppen{-1000}%
    \def\@cite##1##2{(##1\if@tempswa , ##2\fi)}%
    \def\citeauthoryear##1##2##3{##3}\@citedata}
\def\citeyearNP{\def\@citeseppen{-1000}%
    \def\@cite##1##2{##1\if@tempswa , ##2\fi}%
    \def\citeauthoryear##1##2##3{##3}\@citedata}
%
%
%
\def\@citedata{%
    \@ifnextchar [{\@tempswatrue\@citedatax}%
                  {\@tempswafalse\@citedatax[]}%
}

\def\@citedatax[#1]#2{%
\if@filesw\immediate\write\@auxout{\string\citation{#2}}\fi%
  \def\@citea{}\@cite{\@for\@citeb:=#2\do%
    {\@citea\def\@citea{, }\@ifundefined
       {b@\@citeb}{{\bf ?}%
       \@warning{Citation `\@citeb' on page \thepage \space undefined}}%
{\csname b@\@citeb\endcsname}}}{#1}}%

%
\def\@citex[#1]#2{%
\if@filesw\immediate\write\@auxout{\string\citation{#2}}\fi%
  \def\@citea{}\@cite{\@for\@citeb:=#2\do%
    {\@citea\def\@citea{; }\@ifundefined
       {b@\@citeb}{{\bf ?}%
       \@warning{Citation `\@citeb' on page \thepage \space undefined}}%
{\csname b@\@citeb\endcsname}}}{#1}}%

%
\def\@biblabel#1{}
\makeatother



\newdimen\bibindent
\bibindent=0.0em
\def\thebibliography#1{\section*{\refname}\list
   {}{\settowidth\labelwidth{[#1]}
   \leftmargin\parindent
   \itemindent -\parindent
   \listparindent \itemindent
   \itemsep 0pt
   \parsep 0pt}
   \def\newblock{}
   \sloppy
   \sfcode`\.=1000\relax}


\setlength{\baselineskip}{12.7pt}

\title{FAST MONTE CARLO}

\author{\begin{center}Irene Aldridge\textsuperscript{1}\\
[11pt]
\textsuperscript{1}Cornell University, ORIE, Financial Engineering, Ithaca, NY, USA
\end{center}}

\maketitle

\vspace{-12pt}

\section*{ABSTRACT}
This paper proposes an eigenvalue-based small-sample approximation of the celebrated Markov Chain Monte Carlo that delivers an invariant steady-state distribution that is consistent with traditional Monte Carlo methods. The proposed eigenvalue-based methodology reduces the number of paths required for Monte Carlo from as many as 1,000,000 to as few as 10 (depending on the simulation time horizon $T$), and delivers comparable, distributionally robust results, as measured by the Wasserstein distance. The proposed methodology also produces a significant variance reduction in the steady-state distribution.

\section{INTRODUCTION}

Monte Carlo simulation is an intuitive computational methodology that estimates the distribution of a given variable. Monte Carlo is flexible and non-parametric, repeatedly sampling various possible evolution paths of the variable instead of imposing rigid modeling specifications. Monte Carlo is ubiquitous in applications, such as asset pricing: see, for example, \shortcite{BOYLE19971267} and \shortcite{Glasserman2003}, forecasting: \shortcite{Leith1974_TheoreticalSkillofMonteCarloForecasts}, and, most recently, machine learning optimization: \shortcite{PainterEtAl2023} and \shortcite{henderson2020variancereductionsimulationmulticlass}, among many others. The application of Monte Carlo proved to be particularly effective in multi-agent problems, such as multi-period games: \shortcite{cheng:input94} and \shortcite{macfarlane2024sposequentialmontecarlo}). Most recently, \shortcite{AgrawalEtAl2024} use Monte-Carlo to estimate Efficient Influence Functions (EIF) that determine causal relationships among different variables. \shortcite{abbassi2015optimizing} applied Monte Carlo to advertising optimization. Monte Carlo is even used to assess the structural stability of buildings (see \shortcite{SONG2023103479} for a comprehensive survey). 

One of the drawbacks of Monte Carlo algorithms is that they take a long time to compute. For example, \shortcite{AndersonHighamSun2018} identify computational complexities for different versions of Monte Carlo implementations ranging from $\Theta(N^{2\alpha} + N)$ to $\Theta(N^{2\alpha-1}(log\: N)^2+N)$. Different mitigation strategies have been proposed to speed up the computation time, including using advanced hardware by \shortcite{Rosenthal2000} and machine-learning-aided estimation by \shortcite{LamZhang2023} and \shortcite{liCapponi2025predictionenhancedmontecarlomachine}. 

In this paper, I propose to dramatically speed up the computation of Monte-Carlo using the Perron-Frobenius theorem. Extending the classic Markov Chain Monte Carlo (MCMC) due to \shortcite{Metropolis1953} and \shortcite{Hastings1970}, the proposed algorithm:
\begin{enumerate}
    \item incorporates intermediate simulation steps, discarded in classic Monte-Carlo approaches, into the Markov Chain. This part of the process can be calculated in $\mathcal{O}(N)$.
    \item applies Perron-Frobenius theorem to find the steady-state distribution of the Markov Chain, that also happens to be the terminal Monte-Carlo distribution. This part of the process can be implemented in constant time $\mathcal{O}(1)$
    \item the complete process runs in $\mathcal{O}(N)$, significantly faster than comparable alternatives.  
\end{enumerate} 

In the classic Markov Chain Monte Carlo (MCMC) framework developed by \shortcite{VonNeumann1945}, \shortcite{Metropolis1953} and \shortcite{Hastings1970}, the Markov chain records transitions from the beginning state to the terminal state of each simulation path: $S_0$ to $S_T$. The classic MCMC ignores intermediate steps or subpaths from $S_t$ to $S_{t+1}$ for $0<t<T-1$.  However, the intermediate subpaths are all drawn from the same distribution and, therefore, are valid paths under each Monte Carlo specification. Further, each subpath forms a path toward the terminal state destination. I propose to include the subpaths in the Markov matrix to dramatically increase the number of observations in the Markov Chain. The computational complexity of the resulting process is $\mathcal{O}(N)$, directly proportional to $N$, the number of paths chosen for the operation. This methodology further develops the Stein approach extended by \shortcite{LamZhang2023}. 

Next, traditional MCMC approaches rely on computationally complex algorithms to find the steady-state terminal distribution of the simulation.  I propose to use the Perron-Frobenius theorem to quickly estimate the steady-state distribution as the first eigenvector of the Markov Chain. With recent advances in matrix multiplication, the first eigenvector can be computed in $\mathcal{O}(1)$. 

Taken together, the extended Markov chain and the Perron-Frobenius estimation of the steady state result in $\mathcal{O}(N)$ complexity of the Monte Carlo calculation. In addition to the theoretical derivation, I provide simulation results that vividly illustrate the process. 

The proposed eigenvalue framework leverages low-rank approximations. Related approaches include \shortcite{kozyrskiy2020cnnaccelerationlowrankapproximation} and \shortcite{hu2021loralowrankadaptationlarge}, who use low-rank decomposition to increase the computational speed of Convolutional Neural Networks (CNNs) and Large Language Models (LLMs), respectively.

This work falls under the umbrella of distributionally robust optimization (DRO). The field of DRO was pioneered by \shortcite{Scarf1958}. This paper further advances ideas proposed by \shortcite{CalafioreElGhaoui2006}, \shortcite{BertsimasJohnsonKallus2015} and \shortcite{BertsimasGuptaKallus2018} who emphasized the importance of optimization over randomization. Specifically, \shortcite{BertsimasGuptaKallus2018} showed that small sample inferences can be robustly estimated using \textit{some} statistical test. We follow this approach to show that the proposed Monte Carlo methodology also delivers statistically robust inferences.  

I measure convergence using Wasserstein distance, which is identical to the MC convergence metric of \shortcite{bellin1994eulerian}. Furthermore, traditional Monte Carlo methods converge in $\mathcal{O}(\sqrt{N})$. We show that our methodology converges in $\mathcal{O}(1)$.

Last but not least, the proposed approach significantly reduces the variance of Monte-Carlo result. \shortcite{liCapponi2025predictionenhancedmontecarlomachine} document the importance of variance reduction in real-life applications. The proposed methodology delivers both accurate estimates and a considerably smaller variance of the result.

\section{Related Literature}
\subsection{Monte-Carlo}

The original idea of the Monte-Carlo framework was developed by \shortcite{VonNeumann1945}, \shortcite{metropolis_ulam_1949} and \shortcite{Metropolis1953} as part of the Manhattan project. Monte-Carlo referred to the mathematical modeling of random movements of particles, and was named after the popular gambling destination of that era. Subsequently, Monte-Carlo became a go-to tool for integration and other computational problems popularized by \shortcite{Bauer1958}, \shortcite{HammersleyMorton1954}. Manufacturing, production and scheduling applications soon followed in \shortcite{YouleEtAl1959}, \shortcite{Jessop1956}, \shortcite{CraneEtAl1955}, \shortcite{Miller1961}, \shortcite{Blumstein1957}, \shortcite{Rich1955} and surveyed by \shortcite{Shubik1960}.  \shortcite{Johnson2013} notes that Monte-Carlo rapidly expanded to social sciences as well, including \shortcite{McPheeSmith1962Voting} and \shortcite{Barnett1962}. Since then, Monte-Carlo has been widely used in finding distributionally-stable statistical estimators, like in \shortcite{DAgostinoRosman1974} and many others. Monte-Carlo found a particularly enthusiastic following in Finance: \shortcite{PALMER1994264}, \shortcite{Johnson2002} and \shortcite{LinnTay2007}. 
 
The most recent applications of Monte Carlo fall into the following three, equally important, classes:
\begin{itemize}
    \item Diffusion
    \item Agent-based modeling
    \item Reinforcement Learning
\end{itemize}

\subsubsection{Monte-Carlo in Diffusion}

Monte Carlo diffusion extends the standard Brownian motion to find the distribution of future values. For example, in Finance, Monte Carlo diffusion is often used to price options by first simulating the distribution of the underlying variable (e.g., a stock) and then estimating the option payout for each outcome and taking the present value of the payouts to find the fair option price. A sample stock price diffusion model and the option payout calculation shown in equations 

\begin{equation}
    dS_t = \mu S_t dt + \sigma S_t dW_t
\end{equation}
\begin{equation}
    P_{C}(S_t) = max(S_t-K,0)
\end{equation}
\begin{equation}
    C_t = E_{\theta}[P_c(S_t)]
\end{equation}
where $E_{\theta}$ is the expected value under the chosen probability measure.

\shortcite{Giles2008} proposed Euler discretization which he called MultiLevel Monte Carlo (MLMC). The idea of MLMC is to cut down on computation by generating a few high-accuracy approximations and many low-accuracy approximations and sampling repeatedly from the combined pool of paths. The fastest variation of MLMC is the unbiased MLMC tau-leaping algorithm developed by \shortcite{AndersonHigham2012}.

\subsubsection{Monte-Carlo in Agent-based Models}

Monte Carlo has been used in agent-based models, where simulation is used to determine the likely outcomes of the actions of loosely connected autonomous agents.  This methodology estimates the joint likelihood function of many potentially heterogeneous agents with possibly unobserved distributions.

Sequential Monte Carlo (SMC) has been deployed to estimate the densities $p(x_t|y_t,\theta)$ of unknown variables $x_t,\;t=1,...,T$ given the observed variables $y_t$, $t=1,...,T$, and parameters $\theta$. SMC provides a point-wise approximation of density $p(x_t|y_t,\theta)$ where the closed-form solution may not be available. Due to their point-wise nature, SMC is often referred to as a particle filter method and is popular in the estimation of the joint likelihood function. The seminal works in this area include \shortcite{GordonEtAl1993} and \shortcite{kitagawa1996monte}. Since estimation can be based on previously-observed variables, the SMC methodology is also considered to be Bayesian.  

In a system with $B$ agents or particles, a random collection of $j=1,...,B$ particles is sampled from the stationary distribution of the hidden variable $x_t$. The densities $P(y_t|x_t^{(j)})$ are computed. Next. particles are resampled using normalized weights $P(y_t|x_t^{(j)})/\sum_{j=1}^B P(y_t|x_t^{(j)})$. A one-period simulation is conducted on each particle to bring it to $t+1$ realization. The process is repeated for $t=1,...,T$.  Recent applications of SMC include \shortcite{amisano2010euro}, \shortcite{fernandez2007estimating}, \shortcite{bao2012particle}, \shortcite{pitt2014simulated}, \shortcite{blevins2016sequential} and \shortcite{gallant2018bayesian}. 

The SMC process is computationally-intensive, requiring $\mathcal{O}(B\times T\times  (\text{rounds of optimization}))$. The methodology proposed in this paper significantly reduces the computational complexity of the SMC method by reducing the number of optimization rounds required.

\subsubsection{Monte-Carlo in Reinforcement Learning}

A reinforcement learning algorithm can be represented as a sequential decision process. To do so, \shortcite{macfarlane2024sposequentialmontecarlo} use the Markov Decision Process (MDP) framework, a standard in Q-learning.  MDP can be described as a tuple $(\mathcal{S}, \mathcal{A},\mathcal{T}, r, \gamma, \mu)$, where, according to \shortcite{macfarlane2024sposequentialmontecarlo}:
\begin{itemize}
    \item $\mathcal{S}$ is the set of all possible states
    \item $\mathcal{A}$ is the set of actions an agent can take
    \item $\mathcal{T}$ is the state transition probability function, $\mathcal{T}: \mathcal{S}\times\mathcal{A}\to\mathcal{P}(\mathcal{S})$
    \item $r$ is the reward function: $r: \mathcal{S}\times\mathcal{A}\to\mathcal{R}$
    \item $\gamma$ is the discount factor, $\gamma\in[0,1]$
    \item $\mu$ is the initial state distribution. 
\end{itemize}

 The sequential decision process also incorporates a set of agent actions following a policy $\pi: \mathcal{S}\to\mathcal{P}(\mathcal{A})$. At each step $t$, an agent:
 \begin{enumerate}
     \item chooses action $a_t\in\mathcal{A}$
     \item ends up in state $s_t$
     \item receives a reward $r_t = r(s_t, a_t)$
 \end{enumerate}

 At time $t=0$, the present value of the agent's reward from step $t$ is $\gamma^t r_t$.  Over the lifetime of the process, the present value of the cumulative rewards at time $t=0$ is then $\sum_{t=0}^{\infty}\gamma^t r_t$. The agent seeks an optimal policy $\pi^*$ to maximize this amount:
 
 \begin{equation}
     \pi^*\in arg \:max_{\pi\in\Pi} E_\pi \left[\sum_{t=0}^{\infty}\gamma^t r_t\right]
 \end{equation}
where $\Pi$ is a set of all realizable policies. 

We also define the value function:
\begin{equation}
    V^\pi(s_t) = E_\pi \left[\sum_{t=t}^{\infty}\gamma^t r_t|s_t\right]
\end{equation}
which is the present value of future rewards when the agent starts out from the state $s_t$ and follows the policy $\pi$.

The state-action value function maps a state $s_t$ to the present value of future rewards when first choosing action $a_t$ and then following policy $\pi$:
\begin{equation}
    Q^\pi(s_t, a_t) = E_\pi \left[\sum_{t=t}^{\infty}\gamma^t r_t|s_t, a_t\right]
\end{equation}

\shortcite{ToussaintStorkey2006}, \shortcite{KappenEtAl2012} and others formulate the distribution over possible paths $\tau=(s_0,a_0, s_1,a_1,...,s_T,a_T)$ for horizon $T <\infty$:
\begin{equation}
    p(\tau)=\mu(s_0)\Pi_{t=0}^Tp(a_t)\mathcal{T}(s_{t+1}|s_t,a_t)
\end{equation}
where $\mathcal{T}$ is the transition probability matrix defined above and
\begin{itemize}
    \item $\mu_0$ is the initial state distribution 
    \item $a_t$ is the immediately preceding action 
\end{itemize}

The Monte Carlo process is used to estimate the optimal policy $\pi^*$. This is accomplished by:
\begin{enumerate}
    \item Sampling possible paths from \textbf{$\tau$}
    \item Simulating terminal outcomes
    \item Computing present values of the outcomes under different specifications (such as probabilities of individual states)
    \item Finally, choosing the path with the highest present value.
\end{enumerate}

\shortcite{AgrawalEtAl2024} use Monte-Carlo reinforcement learning to estimate Efficient Influence Functions (EIF) that determine causal relationships among different variables.

\section{PROPOSED METHODOLOGY}

Following \shortcite{ToussaintStorkey2006} and \shortcite{KappenEtAl2012}, we define a path $\mathcal{\tau}$ as a sequence of state transitions under the Monte-Carlo simulation from the beginning state to the terminal state:
\begin{equation}
    \tau=(s_0, s_1, ,...,s_T)
\end{equation}
In the methodology discussed in this section, we explicitly omit the action space required in reinforcement learning, however, it can be inserted without sacrificing the results.

\subsection{Markov Chain Construction}

The eigenvalue-decomposition approach for Monte Carlo presented in this paper further optimizes any Monte Carlo methodology to reduce the required number of paths. Our approach is related to the Metropolis-Hastings formulation (\shortcite{MetropolisEtAl1953}, \shortcite{Hastings1970}), but also relies on the Perron-Frobenius theorem (\shortcite{Frobenius1912}).  

As with \shortcite{ToussaintStorkey2006} and others, each simulation randomly samples a path from a set of all possible paths. Then, following Metropolis-Hastings Markov Chain Monte Carlo (MCMC), using the sample paths, we construct a transition probability matrix. 

The proposed Markov Chain construction is different from the Metropolis-Hastings MCMC. The Metropolis-Hastings approach only considers transitions from the first step to the respective terminal step, $s_0 \to s_T$, effectively discarding all the simulation steps in between. The approach proposed in this paper records every transition step in the Markov Matrix: $s_0\to s_1, s_1\to s_2,\dots, s_{t-1}\to s_t\;\forall t\in\{1,\dots,T\}$. Although the recorded transitions are much smaller than the start-to-end transitions under Metropolis-Hastings MCMC, when considered in the steady-state framework, both the proposed transitions and Metropolis-Hastings MCMC converge to the same terminal distribution. This result follows directly from the principles of Markov Chain construction. By construction, the proposed Eigen Monte Carlo approach produces a much larger sample of transitions than does Metropolis-Hastings MCMC, requiring a lot fewer paths to draw sound steady-state conclusions. 

Our proposed framework, therefore, increases the number of observations in the Markov Chain without increasing the number of simulation paths. 

\subsection{Computational Improvement of Steady-State Calculations}

Next, we further improve the computation speed of the steady-state by deploying the Perron-Frobenius theorem: the first eigenvector of a Markov Chain transition probability matrix serves as the vector of the steady-state probabilities of the Markov Chain matrix. With Perron-Frobenius, we are able to deduce a steady-state distribution of the full Monte Carlo simulation with fewer steps. 

To compute the steady-state vector for the transition probability matrix $\mathcal{T}$:
\begin{enumerate}
    \item Find the largest eigenvalue $\lambda_{max}$ and the corresponding singular vector $v_0$ by solving $(A - I_n)v_0=0$
    \item Divide $v_0$ by the sum of entries $\sum_i v_0$ to obtain a normalized vector $w$ that sums up to 1.
    \item The vector $w$ will automatically have all positive values. This is the steady-state vector for the transition probability matrix $\mathcal{T}$. 
\end{enumerate}

It is easy to prove that the largest eigenvalue of a Markov matrix is always 1 (\shortcite{Seabrook_2023}). 

Our key result is to show that the Perron-Frobenius eigenvector $v_0$ corresponding to $\lambda=1$ in any Markov matrix is always invariant. That is, regardless of how many observations we have, the vector $v_0$ retains its distributional properties and represents steady-state distributions of the Markov matrix.  

\begin{theorem}
    The steady-state distribution of a Markov Chain Monte Carlo is invariant for any number of paths $k$ which satisfies a normal distribution.  
\end{theorem}
\begin{proof}
This result follows directly from the eigenvalue definition applied to the transition probability matrix $\mathcal{T}$:
\begin{equation}
    \mathcal{T}v_0 = \lambda_{max} v_0
\end{equation}
Since $\lambda_{max} = 1$ for any Markov matrix, 
\begin{equation}\label{eq:MC1}
    \mathcal{T}v_0 = v_0
\end{equation}

Suppose now that the Markov transition probability matrix is perturbed:
\begin{equation}
    \mathcal{T'} = \mathcal{T} + \epsilon
\end{equation}
Then, $v_0'$ is the first eigenvector of $\mathcal{T'}$. Furthermore, since equation (\ref{eq:MC1}) holds for all Markov matrices:
\begin{equation}\label{eq:T'}
    (\mathcal{T}+\epsilon)v_0' = v_0'
\end{equation}

Equation (\ref{eq:T'}) can be further represented as:
\begin{equation}
    \mathcal{T}v_0'+\epsilon v_0' = v_0'
\end{equation}

Since, equation (\ref{eq:MC1}),  $\mathcal{T}v_0=v_0$, holds for any $v_0$, including $v_0'$, $\epsilon v_0'\to 0$.

The $v_0'$, therefore, remains invariant from the very few original paths taken by the Monte Carlo simulation.     
\end{proof}

Empirically, as discussed in the next section, I find that just a handful of the paths are enough to produce inferences comparable with full-fledged Monte Carlo.

The proposed algorithm works as follows:
\begin{algorithm}
  \caption{Fast Monte Carlo} \label{alg:my-algorithm}
  \SetKwInOut{Input}{inputs}
  \SetKwInOut{Output}{output}
  \SetKwProg{FastMonteCarlo}{FastMonteCarlo}{}{}

  \FastMonteCarlo{}{
    \Output{Output of a Monte Carlo Simulation}
    initialize a Markov chain $M(s_i,s_j) = 0$\;
    initialize a Monte-Carlo process\;
    \While{$i \geq N$}
    {  
        \ForEach {$\text{path}\; \tau_i=(s_{i0}, s_{i1}, ,...,s_{iT})$}
        {
            Record each intertemporal transition in the Markov chain (e.g., record transition from each state $s_t$ to state $s_{t+1}$ for $t=1, 2, ...T$: $M(s_t, s_{t+1}) \gets +1 \;\forall \;t$
        }
    }
    Run eigenvalue decomposition on the Markov Chain $M$\;
    By Perron-Frobenius theorem, the first eigenvector $V_0$, normalized, represents steady-state probabilities associated with matrix $M$\;
    \KwRet{$V_0$}\;
  }
\end{algorithm}

The resulting distribution of the states and their probabilities match those of full Monte Carlo at a fraction of the computational cost. Note that since the Markov Chain incorporates intra-path transitions, the proposed process can run with a fraction of paths planned: $n<<N$.

\section{COMPUTATIONAL COMPLEXITY}

\shortcite{AndersenEtAl2006} and \shortcite{SunEtAl2020} showed that the steady state of a Markov transition probability matrix can be computed in nearly constant time. The result is at least in part due to fast matrix multiplication algorithms, which are themselves an active area of research in reinforcement learning (see \shortcite{FawziEtAl2022}). 

As a result, the Monte-Carlo simulation can be distilled in as few as $\mathcal{O}(N)$ steps. The resulting convergence rate of Monte-Carlo is $\mathcal{O}(1)$, independent of the number of steps due to the invariance of the first eigenvalue and the related eigenvector. 

\section{EXPERIMENTS}

\subsection{Standard Monte-Carlo}
In this section, we reproduce a standard Monte Carlo simulation and compare the speed of the proposed approach. Specifically, we represent the simulated data in a Markov matrix and then find the steady-state distribution using the Perron-Frobenius matrix. Finally, we compare the resulting distribution with a base Monte Carlo distribution obtained by simulating 1,000,000 paths. 

In a traditional Monte Carlo simulation, we seek to record the terminal result of each path and create as many paths as possible. In contrast, the proposed method studies the intermediate steps of the path itself. In the proposed model, we examine the entire data path and treat each step as an independent transition that we can subsequently use in the analysis. 

In the simulation, we construct a basic binomial tree and compare the performance of the traditional 1,000,000 path Monte Carlo model with the proposed eigenvalue-based model. 

Each of the 1,000,000 paths was constructed as follows:
\begin{itemize}
    \item Each path started with the base value of variable $S_0=1$
    \item A $21\times 21$ Markov Transition Probability matrix $\mathcal{T}$ was initialized. Each row corresponded to 0.1 increment in value. The row indexed by $i=11$ corresponded to $S=1$, $i=12$: $S=1.1$, etc.   
    \item In each subsequent step $t\in\{1,...9\}$, a random number $r_t\in[0,1]$ was drawn. Whenever $r_t=1$, the value of the variable $S$ increased by 10\% ($\Delta s_u =1.1$): 
    \begin{equation}
        S_t|_{r_t=1} = S_{t-1} \Delta s_u
    \end{equation} 
    Similarly, whenever $r_t=0$, the value of the variable $S$ decreased by 10\% ($\Delta s_d = 0.9$): 
    \begin{equation}
        S_t|_{r_t=0} = S_{t-1} \Delta s_d
    \end{equation}

    \item Along each path, all transitions on the path were added to the Markov Transition Probability matrix $\mathcal{T}$. For example, if $S_{t-1}=1$ and $S_t=1.1$, then $\mathcal{T}_{11,12}$ was incremented by 1.   
    \item At the end of each path, a copy of the matrix $\mathcal{T}$ was created and normalized. The steady-state probabilities were estimated as the first eigenvector, normalized. The resulting steady-state distribution corresponding to the path $n\in 1,000,000$ was recorded.
    \item When $n=1,000,000$ paths were simulated, the average steady-state "true Monte Carlo" distribution was calculated.
    \item For each value of $n\in 1,000,000$, the Wasserstein Distance was computed between the eigen-distribution ($d_e$) corresponding to the first $n$ paths and the true Monte-Carlo distribution ($d_{MC}$). The Wasserstein distance was calculated as:
    \begin{equation}
        W_p = E[(d_e - d_{MC})^2]^{1/2}
    \end{equation}
\end{itemize}

    Figures \ref{fig:raw-wasserstein-binomial} and \ref{fig:stylized-wasserstein-binomial} summarize the results. As Figure \ref{fig:raw-wasserstein-binomial} shows, the Wasserstein distance between the Eigen Monte Carlo and the true Monte-Carlo is minimal even when the Eigen Monte Carlo is estimated only   on the first 10 paths! This confirms our theoretical results. 
    
    \begin{figure}
        \centering
        \includegraphics[width=0.5\linewidth]{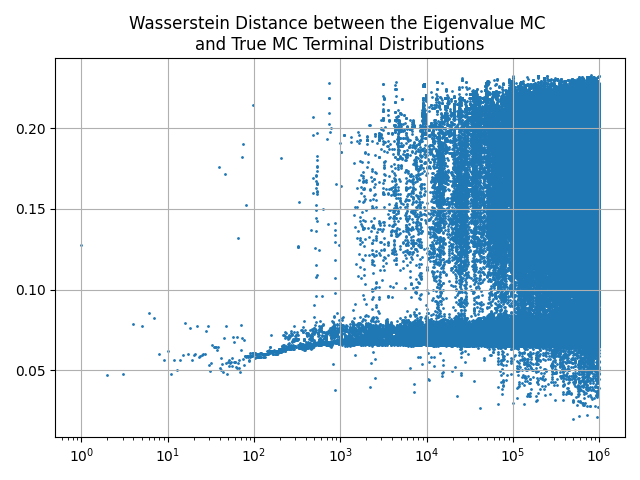}
        \caption{Raw Wasserstein Distance Between the 1,000,000 path Monte-Carlo and the Eigen Monte Carlo determined on the first $x$ paths.}
        \label{fig:raw-wasserstein-binomial}
    \end{figure}
    
    \begin{figure}
        \centering
        \includegraphics[width=0.5\linewidth]{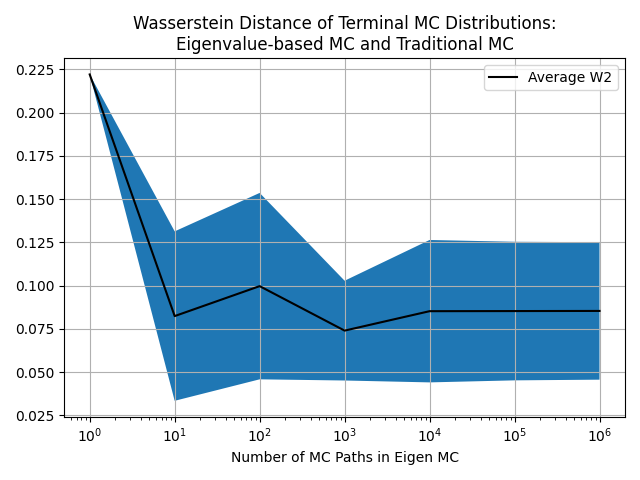}
        \caption{Mean Wasserstein Distance +/- 1 Standard Deviation between the 1,000,000 path Monte-Carlo and the Eigen Monte Carlo determined on the first $x$ paths.}
        \label{fig:stylized-wasserstein-binomial}
    \end{figure}

The simulation performed generalizes to Brownian motion and other Monte Carlo-based techniques. Given recent advances in calculating the first singular vector (see, for example, \shortcite{AndersenEtAl2006} and \shortcite{SunEtAl2020}), the proposed Eigen Monte Carlo approach can run in as little as $O(1+\epsilon)$ computational time. This is a significant improvement over existing methods.

\subsection{Diffusion, Step=by-Step}\label{diffusion-simulation}

This section compares the traditional and proposed Monte Carlo approaches in options pricing. We seek to estimate a European call option on a hypothetical stock with a mean drift $\mu=0.002$, volatility $\sigma=0.01$, starting price $S_0=100$ and option strike price of $110$. To keep the simulation tractable and easy to replicate, we simulate $N=300$ paths over a $T=30$ time horizon.

We simulate the diffusion process $dS/S = \mu dt + \sigma dW_t$, where $W_t$ is a standard Brownian motion. The resulting lognormal process translates to 
\begin{equation}
S_t = S_{t-1}exp\left((\mu + \sigma^2/2) dt + \sigma dW\right)  
\end{equation}

The continuous-value setting of traditional diffusion is a challenge to the MCMC approach, because MCMC enforces discretization of the states. In our experiment, to maintain tractability, we split all possible diffusion values into just 40 states, confining the possible price levels to whole values only, in the range $S=[80,120]$. If the simulation breached either of the boundary values, the algorithm bucketed the values into the boundary state, another limitation of the experiment. Thus, a simulated price of $80$ would be recorded as $80$, but so would the price of $79$ and any lower price. 

To compare the evolution of $S$ under the traditional and proposed MCMC approaches, we construct two distribution:
\begin{enumerate}
    \item A classic MC distribution of $N$ full path transitions from state $S_0$ to $S_T$.
    \item The "Eigen Monte Carlo" distribution constructed using the approach proposed in this paper. To develop this distribution, we first construct a Markov Chain that incorporates all intermediate transitions from each state $S_i$ to $S_{i+1}$, where $i<T$. In this example, the Markov Chain contained just 40 rows and 40 columns. Next, we normalize each row and discard the rows where all the elements were zeros to simplify the estimation process. Finally, we performed the SVD, obtained the first singular vector, and normalized it to find the steady-state distribution. 
\end{enumerate}

\begin{figure}
    \centering
        \centering
        \includegraphics[width=0.45\linewidth]{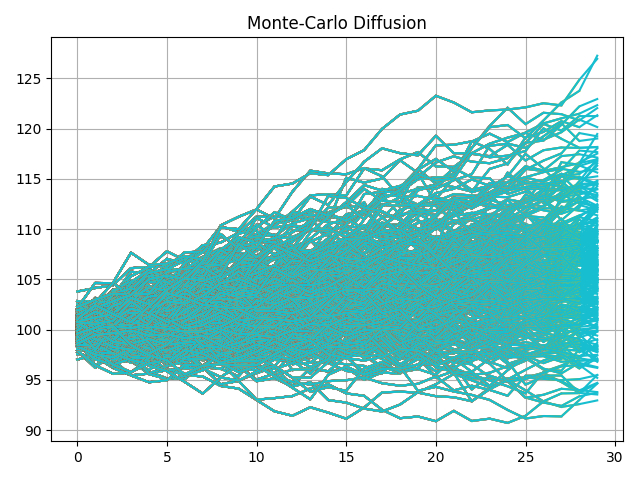}
        \caption{Diffusion}
        \label{fig:MCDiffusion}
\end{figure}

\begin{figure}
    \centering
    \begin{minipage}{0.49\textwidth}
        \centering
        \includegraphics[width=0.99\linewidth]{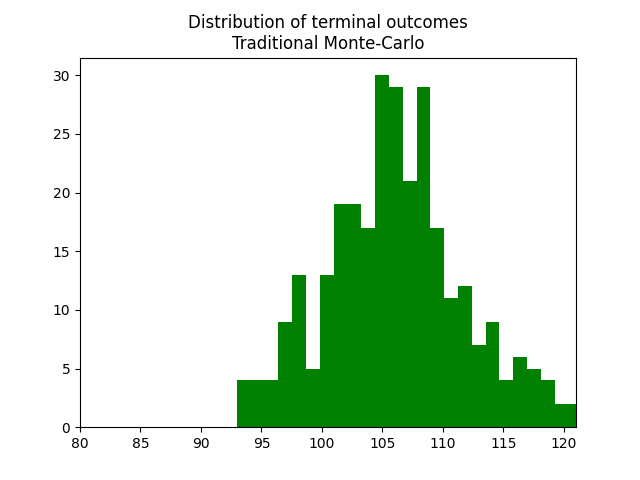}
        \caption{Traditional Diffusion: Terminal State.}
        \label{fig:TradDiffusion}
    \end{minipage}
    \begin{minipage}{0.49\textwidth}
        \centering
        \includegraphics[width=0.99\linewidth]{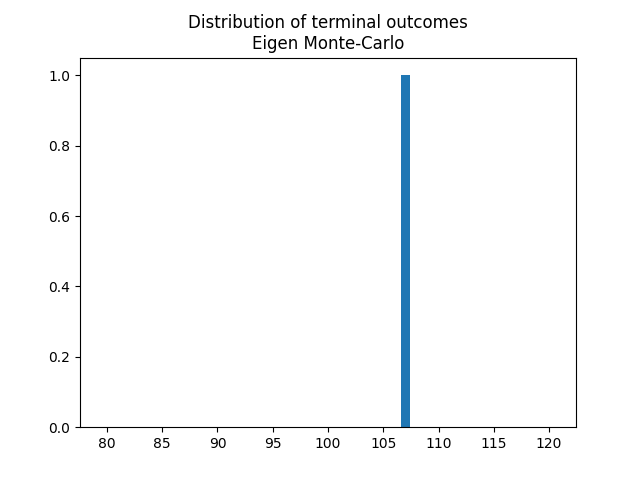}
        \caption{Proposed Eigen-MC: Terminal State.}
        \label{fig:EigenDiffusion}
    \end{minipage}
\end{figure}

Figure \ref{fig:MCDiffusion} shows the $N=300$ simulated diffusion paths of a toy example with mean drift $\mu=0.002$, volatility $\sigma=0.01$, starting price $S_0=100$ and option strike price $110$. Figures \ref{fig:TradDiffusion} and \ref{fig:EigenDiffusion} show the terminal distributions of the simulated stock prices for the traditional diffusion Monte-Carlo and the proposed approach. The proposed Eigen Monte Carlo delivers a small-variance estimate.

However, the proposed approach did not always converge in the diffusion setting. The lack of convergence may be driven by excessive discretization. Section \ref{simulation-parameters} considers some aspects of the approach performance driven by the model parameterization.

\section{Discussion of Results}\label{simulation-parameters}

The proposed methodology increases the underlying sample space of the Monte-Carlo simulation, without increasing the computational complexity of the process. Markov Chains effectively document the flow of data among states. The steady-state analysis captures the end points of the information flow, whether the data is flowing through the intermediate states or going directly to a terminal state. By recording all intermediate steps in the simulation, we increase the sample size and deliver more robust inferences with smaller variance. 

Much further work is needed to determine the conditions and parameters that determine the simulation outcome:
\begin{itemize}
    \item The impact of the number of paths $N$ and the number of time steps $T$  required to derive stable inferences using the Eigen Monte Carlo methodology proposed in this paper.
    \item The effect of the discretization of the Markov matrix: how granular should the states be and what is the impact of approximating the state values in the Markov chain?
    \item The parameters of the simulation itself: for example, the relationship between the drift and the variance. 
\end{itemize}

\section{CONCLUSION}

This paper develops and tests a new approach to the classic Monte Carlo simulation. Recording all intra-path transitions in a Markov Transition matrix and then using Eigenvalue decomposition and related Perron-Frobenius theorem, we are able to compute steady-state Monte Carlo inferences with 1) a much smaller variance and 2) in just a few sample paths, down from as many as 1,000,000 iterations. The Eigenvalue-based Monte Carlo can be computed in $O(1)$ computational time. 

\footnotesize

\bibliographystyle{wsc}
\bibliography{MonteCarlo,RL,RandomMatrixTheory}

\section*{AUTHOR BIOGRAPHIES}

\noindent {\bf \MakeUppercase{Irene Aldridge}} is a Visiting Professor at Cornell University, ORIE, Financial Engineering. She is a recognized researcher in the applications of Data Science to Finance. Her research interests span advances in financial technology and machine learning/AI, including blockchain-related optimization. Irene is a founder of several Fintech companies, most recently AbleMarkets and AbleBlox, where she served as CEO. In addition to startups, Irene also held several executive roles in banking. She is the author of several books on financial technology, including "High-Frequency Trading" (Wiley, 2013), "Real-Time Risk" (with Steven Krawciw, Wiley, 2017) and "Big Data Science in Finance" (with Marco Avellaneda, 2021).  Her email address is \email{irene.aldridge@gmail.com} and her website is \url{https://irenealdridge.com/}.\\

\end{document}